\documentclass{amsart}%
\usepackage{amsmath}
\usepackage{amssymb}
\usepackage{amsthm}
\usepackage{amscd}
\usepackage[dvips]{graphicx}
\usepackage[mathscr]{eucal}
\newtheorem{Thm}{Theorem}[section]
\theoremstyle{definition}
\newtheorem{Theorem}[Thm]{Theorem}
\newtheorem{Lemma}[Thm]{Lemma}

\theoremstyle{remark}

\font\ym=msbm10

\newcommand{\R}{\text{\ym R}}

\newcommand{\C}{\text{\ym C}}

\newcommand{\sS}{\mathscr S}

\title[]
{Geometry of Coherent States of CCR Algebras}
\author[Shigeru Yamagami]{}

\begin{document}
\maketitle   
\begin{center}
YAMAGAMI Shigeru\footnote{Partially supported by KAKENHI(22540217)}
\end{center}
\begin{center}
Graduate School of Mathematics
\end{center}
\begin{center}
Nagoya University 
\end{center}
\begin{center} 
Nagoya, 464-8602, JAPAN 
\end{center}    
\begin{center}
\email{yamagami@math.nagoya-u.ac.jp}
\end{center}
\subjclass[2000]{46L60, 46L51}
\begin{abstract}
Geometric positions of square roots of coherent states 
of CCR algebras are investigated 
along with an explicit formula for transition amplitudes among them, 
which is a natural extension of our previous results on quasifree states 
and will provide a new insight into quasi-equivalence problems of quasifree states. 
\end{abstract}


\section*{Introduction}
Coherent states are most extensively studied in the case 
of finite degree of freedom but most of its formality 
works in the infinite dimensional case as well, 
which is especially useful in extracting behaviour of classical 
fields from quantum ones. 

We here study coherent states in the form of 
shifted quasifree states as a continuation of 
our previous article \cite{GQFS}. 

A quasifree state $\varphi_S$ is a special state of 
a CCR C*-algebra parametrized by its covariance 
form $S$. Let $\mu_S$ be the Gaussian measure of covariance 
form given by $A = (S^{1/2} + {\overline S}^{1/2})^2/2$. 

Let $\alpha$ be a linear function to be used to 
shift the measure $\mu_S$ with the shifted measure 
denoted by $\mu_{S,\alpha}$, which is the Gaussian measure 
of covariance form $A$ and the mean functional $\alpha$. 
The shift operation is also applicable to quasifree states 
with the shifted state denoted by $\varphi_{S,\alpha}$. 
Our main result is then an explicit formula 
for the transition amplitude 
$(\varphi_{S,\alpha}^{1/2}|\varphi_{T,\beta}^{1/2})$ 
between square roots of coherent states 
$\varphi_{S,\alpha}$ and $\varphi_{T,\beta}$, which is in turn 
equal to the Hellinger integral of $\mu_{S,\alpha}$ and 
$\mu_{T,\beta}$. Thus the geometric position of 
vectors $\{ \varphi_{S,\alpha}^{1/2} \}$ is exactly 
that of classical Gaussian measures $\{ \mu_{S,\alpha} \}$ in the $L^2$-space. 

\section{Preliminaries}
Let $C$ be a C*-algebra. 
Given states $\varphi, \psi$ of $C$, 
let $\varphi^{1/2}$ and $\psi^{1/2}$ 
be their GNS-vectors in the universal representation 
space $L^2(C^{**})$ of $C$, with their inner product 
$(\varphi^{1/2}|\psi^{1/2})$ referred to as the 
transition amplitude. 

By a presymplectic vector space, we shall mean 
a pair $(V,\sigma)$ of a real vector space $V$ and 
an alternating form $\sigma$ on $V$. 
Given a presymplectic vector space $(V,\sigma)$, 
the CCR C*-algebra, denoted by $C^*(V,\sigma)$, 
is the universal C*-algebra generated by symbols 
$\{ e^{ix} \}_{x \in V}$ under the Weyl form of 
CCR's 
\[
e^{ix} e^{iy} = e^{-i\sigma(x,y)/2} e^{i(x+y)}, 
\quad 
(e^{ix})^* = e^{-ix}, 
\quad x, y \in V. 
\]

A positive (sesquilinear) form $S$ on the complexified 
vector space $V^\C$ is called a covariance form on 
$(V,\sigma)$ if 
\[
S(x,y) - {\overline S}(x,y) = i\sigma(x,y) 
\quad
\text{for $x, y \in V^\C$.}
\]
Here ${\overline S}(x,y) = \overline{S( \overline y, 
\overline x)}$. 

Each covariance form $S$ in turn gives rise to a special state 
$\varphi_S$ of the CCR C*-algebra $C^*(V,\sigma)$, 
the quasifree state of $C^*(V,\sigma)$ associated to $S$, by 
\[
\varphi_S(e^{ix}) = e^{-S(x,x)/2}, 
\quad 
x \in V. 
\]

Now the result in \cite{GQFS} is summarized as follows: 
Let $S$ and $T$ be covariance forms on a presymplectic 
vector space $(V,\sigma)$ with 
the associated quasifree states denoted by 
$\varphi_S$ and $\varphi_T$ respectively. 
Then 
\[
(\varphi_S^{1/2}|\varphi_T^{1/2}) 
= \sqrt{ 
\det\left(
\frac{2\sqrt{AB}}{A+B}
\right)}. 
\]
Here positive forms $A, B$ on $V^\C$ are defined by 
\[
A = \frac{1}{2} (S^{1/2} + {\overline S}^{1/2})^2, 
\quad 
B = \frac{1}{2} (T^{1/2} + {\overline T}^{1/2})^2
\]
with the functional calculus on positive forms 
performed in the Pusz-Woronowicz' sense.
We also use the ratio notation $\frac{F}{G}$ 
for two positive forms $F$, $G$ on a complex vector space $K$ 
to stand for an operator satisfying 
$F(x,y) = G(x, \frac{F}{G}y)$ ($x, y \in K$)  
if $F$ is dominated by $G$. Thus 
$\displaystyle \det\left( \frac{2\sqrt{AB}}{A+B} \right)$ 
is the determinant of the operator 
$\displaystyle \frac{2\sqrt{AB}}{A+B}$ fulfilling 
\[
2\sqrt{AB}(x,y) = (A+B)(x, \frac{2\sqrt{AB}}{A+B}y)
\quad 
\text{for $x, y \in V^\C$.}
\]

Related to the above transition amplitude formula, it is 
natural to introduce the following notion: 
Two covariance forms $S, T$ are said to be 
\textbf{HS-equivalent} if the associated $A$ and $B$ are 
equivalent 
(i.e., dominated by each other) with their difference 
represented by a Hilbert-Schmidt class operator 
relative to the hilbertian topology induced from $A$ or $B$.
With this terminology, we have 
$0 \leq \det\left( \frac{2\sqrt{AB}}{A+B} \right) \leq 1$, 
which is strictly positive if and only if 
$S$ and $T$ are HS-equivalent. 

\section{Coherent States}

Let $(V,\sigma)$ be a presymplectic vector space. 
Given a linear functional $\lambda: V \to \R$, 
the replacement $v \mapsto v + \lambda(v)1$ ($v \in V$) 
preserves the CCR's, whence it induces 
a *-automorphism of the CCR C*-algebra $C^*(V,\sigma)$ 
by 
\[
\Phi(e^{ix}) = e^{i\lambda(x)} e^{ix}\quad 
\text{for $x \in V$,}
\]
which, in fact, gives rise to an automorphic action of 
the additive group $V^*$, the algebraic dual of $V$, on 
$C^*(V,\sigma)$. 

If the automorphism $\Phi$ is applied to 
a quasifree state $\varphi_S$, 
then we obtain a \textbf{coherent state} (a quasifree state 
with mean) $\varphi_{S,\lambda}$: $\varphi_{S,\lambda}$ is 
a state of $C^*(V,\sigma)$ specified by 
\[
\varphi_{S,\lambda}(e^{ix}) = \varphi_S(\Phi(e^{ix})) 
= e^{i\lambda(x)} e^{-S(x)/2}, 
\quad 
x \in V. 
\]

\section{Finite-Dimensional Analysis} 
Here 
we shall establish 
a transition amplitude formula between coherent states 
under the assumption that $V$ is finite-dimensional.  

For the moment, we work with a covariance form $S$ such that 
$S + \overline S$ is non-degenerate.
Let 
$\rho_{S,\lambda}$ be the density operator associated to 
the coherent state $\varphi_{S,\lambda}$, 
which is an element in 
the Hilbert algebra $\sS(V,\sigma)$ of 
rapidly descreasing functions on $V$ introduced in 
\cite[\S 3]{GQFS}. 
The gauge automorphism $\Phi$ of $C^*(V,\sigma)$ is 
restricted to  
an automorphism of $\sS(V,\sigma)$ 
so that 
\[
\Phi\left(
\int_V f(x) e^{ix}\,dx 
\right) 
= \int_V f(x) \Phi(e^{ix})\,dx
= \int_V e^{i\lambda(x)} f(x) e^{ix}\,dx, 
\]
namely $(\Phi f)(x) = e^{i\lambda(x)} f(x)$ 
for $f \in \sS(V,\sigma)$. 
Now  
\[
\varphi_{S,\lambda}\left( 
\int_V f(x)e^{ix}\,dx 
\right)
= \varphi_S\left( 
\int_V (\Phi f)(x) e^{ix}\,dx
\right)
= \tau(\rho_S * (\Phi f)) 
= \tau((\Phi^{-1}\rho_S)*f)
\]
shows that $\rho_{S,\lambda} = \Phi^{-1}\rho_S$ and then 
$\rho_{S,\lambda}^{1/2} = \Phi^{-1} \rho_S^{1/2}$; 
\[
\rho_{S,\lambda}^{1/2}(x) 
= \frac{1}{\sqrt{N_S}} 
\exp\left( 
- \frac{1}{4} (S^{1/2} + {\overline S}^{1/2})^2(x) 
- i\lambda(x) 
\right) 
\]
with 
\[
N_S = \int_V e^{-(S^{1/2} + {\overline S}^{1/2})^2(x,x)/2}\, 
dx.
\]

Let $T$ be another covariance form on $(V,\sigma)$ such that 
$T + \overline T$ is non-degenerate. 
Then, by using the density operator expression,  
\begin{align*}
(\varphi_T^{1/2}|\varphi_{S,\lambda}^{1/2}) 
&= \tau(\rho_T^{1/2}*\rho_{S,\lambda}^{1/2}) 
= \int_V \rho_T^{1/2}(-x) \rho_{S,\lambda}^{1/2}(x)\\ 
&= \frac{1}{\sqrt{N_SN_T}} 
\int_V e^{-(A(x)+B(x))/2 + i\lambda(x)}\, dx\\ 
&= \sqrt{\det\left( 
\frac{2\sqrt{AB}}{A+B} 
\right)} 
e^{-(A+B)^{-1}(\lambda)/2}. 
\end{align*}

Recall here that, 
given a positive sesquilinear form $Q$ of a vector 
space $K$, the inverse form $Q^{-1}$ 
(which is a quadratic form on the algebraic dual $K^*$ 
taking values in $[0,+\infty]$) 
is defined as follows: Let $\alpha: K \to \C$. 
If we can find $a \in K_Q$ (the completion of 
$K$ relative to $Q$) satisfying 
$\alpha(x) = Q(a,x)$ for $x \in K$, then $Q^{-1}(\alpha)$ 
is set to be $Q(a,a)$ and otherwise 
$Q^{-1}(\alpha) = +\infty$. 

We shall now remove the non-degeneracy assumption on 
$(\ ,\ )_S = S + \overline S$ and 
$(\ ,\ )_T = T + \overline T$. 

If there is an $x \in V$ such that $(x,x)_T = 0$ 
and $(x,x)_S \not= 0$, then the restrictions 
of $\varphi_T$ and $\varphi_{S,\lambda}$ to $C^*(\R x,0)$ 
are given by a Dirac measure $\delta$ and 
a gaussian measure respectively, whence they are disjoint 
and the determinant formula remains valid. 

Henceforce suppose that $S + \overline S$ and $T + \overline T$ 
are equivalent as positive forms. 
If we can find an $x \in V$ such that 
$(x,x)_S = (x,x)_T = 0$ and $\lambda(x) \not= 0$, then 
the above arugument is applied again to get 
Dirac measures with disjoint supports, 
showing the orthogonality of 
$\varphi_T^{1/2}$ and $\varphi_{S,\lambda}^{1/2}$. 

Finally, consider the case that 
$\lambda(x) = 0$ if $x \in V$ satisfies $(x,x)_S = 0$ 
(or equivalently $(x,x)_T = 0$). Let $V'$ be the 
quotient of $V$ by the kernel of $S+\overline S$ 
with $\lambda'$ the quotient of $\lambda$. 
The quotient map $\phi: V \to V'$ satisfies 
$\lambda = \lambda'\circ \phi$ and then 
it induces a *-homomorphism 
$\pi: C^*(V,\sigma) \to C^*(V',\sigma')$ so that 
$\varphi_{S,\lambda} = \varphi_{S',\lambda'}\circ \pi$ and 
$\varphi_T = \varphi_{T'}\circ \pi$. 
Thus the GNS-representation of $\varphi_{S,\lambda}$ is 
weakly approximated ($\pi$ being an epimorphism) 
by that of $\varphi_{S',\lambda'}$ and 
\cite[Proposition~4.3]{GQFS} is used to see 
\begin{align*}
(\varphi_T^{1/2}| \varphi_{S,\lambda}^{1/2}) 
&= (\varphi_{T'}^{1/2} | \varphi_{S',\lambda'}^{1/2}) 
= \sqrt{\det\left( 
\frac{2\sqrt{A'B'}}{A'+B'} 
\right)} 
e^{-(A'+B')^{-1}(\lambda)/2}\\
&= \sqrt{\det\left( 
\frac{2\sqrt{AB}}{A+B} 
\right)} 
e^{-(A+B)^{-1}(\lambda)/2}
\end{align*}
without assuming the non-degeneracy of 
$S + \overline S$ and $T + \overline T$. 
Note that this particularly implies the inequality
\[
(\varphi_T^{1/2}| \varphi_{S,\lambda}^{1/2}) 
\leq (\varphi_T^{1/2}| \varphi_S^{1/2}).
\]

\section{Infinite-Dimensional Analysis}

Let $S$, $T$ be covariance forms an infinite-dimensional 
presymplectic vector space $(V,\sigma)$ and 
$\lambda: V \to \R$ be a linear functional. 

\begin{Lemma}~ 
\begin{enumerate}
\item 
Unless $S+\overline S$ and $T + \overline T$ are 
equivalent, we have 
$(\varphi_{S,\lambda}^{1/2}| \varphi_T^{1/2}) = 0$. 
\item 
If $S + \overline S$ and $T + \overline T$ are equivalent 
and $\lambda \not= 0$ on 
$\ker(S+\overline S) = \ker(T+\overline T)$, then 
$(\varphi_{S,\lambda}^{1/2}| \varphi_T^{1/2}) = 0$. 
\item 
If $S + \overline S$ and $T + \overline T$ are equivalent 
and $\lambda = 0$ on 
$\ker(S+\overline S) = \ker(T+\overline T)$, then 
\[
(\varphi_{S,\lambda}^{1/2}| \varphi_T^{1/2}) = 
(\varphi_{S',\lambda'}^{1/2}| \varphi_{T'}^{1/2}). 
\]
Here $S'$, $T'$ and $\lambda'$ are induced on the completion of 
quotient space $V' = V/\ker(S+\overline S)$ with respect 
to the inner product $S + \overline S$. 
\end{enumerate}
\end{Lemma}                                
                                

\begin{proof}
(i) Repeat the argument in \cite[\S 4.1]{GQFS} for (i), whereas 
(ii) and (iii) are already discussed 
in the previous section.  
\end{proof}

To check the validity of the transition amplitude formula, 
it therefore suffices to deal with the case where 
$V$ is hilbertian with $S$ and $T$ admissible covariance forms 
(the transition amplitudes being approximated under the process of taking completion, 
see \cite[Proposition~4.3]{GMS}). 

Thanks to the decomposition into seperable subspaces 
(cf.~\cite[\S 4.3]{GQFS}), we can further assume that 
$V$ is a separable hilbertian space. 
Choose an increasing sequence $\{ V_n\}$ of 
finite-dimensional subspaces of $V$ with 
$\cup_n V_n$ dense in $V$. 
Let $S_n$, $T_n$ and $\lambda_n$ be the restrictions of $S$, $T$ 
and $\lambda_n$ to the subspace $V_n^\C$ respectively. 
Set $2A_n = (S_n^{1/2} + {\overline{S_n}}^{1/2})^2$, 
$2A = (S^{1/2} + {\overline S}^{1/2})^2$, 
$2B_n = (T_n^{1/2} + {\overline{T_n}}^{1/2})^2$, 
and $2B = (T^{1/2} + {\overline T}^{1/2})^2$ as before. 

Warning: $A_n$ and $B_n$ are \textit{not} 
necessarily restrictions of $A$ and $B$. 

\begin{Lemma}
We have 
\[
\lim_{n \to \infty} 
\det\left(
\frac{2\sqrt{A_nB_n}}{A_n+B_n}
\right)
= 
\begin{cases}
\det\left( 
\frac{2\sqrt{AB}}{A+B}
\right)
&\text{if $S$ and $T$ are HS-equivalent,}\\
0 &\text{otherwise.}
\end{cases}
\]  
\end{Lemma}

\begin{proof}
This follows from the determinant formula for 
the transition amplitude 
$(\varphi_S^{1/2}|\varphi_T^{1/2})$ 
along with the equality 
$\lim_{n \to \infty} 
(\varphi_{S_n}^{1/2}| 
\varphi_{T_n}^{1/2}) 
= (\varphi_S^{1/2}|\varphi_T^{1/2})$. 
\end{proof}

In view of 
\[
(\varphi_{S_n,\lambda_n}^{1/2}|\varphi_{T_n}^{1/2}) 
= (\varphi_{S_n}^{1/2}|\varphi_{T_n}^{1/2}) 
e^{-(A_n+B_n)^{-1}(\lambda_n)/2},
\]
the above convergence formula gives  
\begin{align*}
(\varphi_{S,\lambda}^{1/2}|\varphi_T^{1/2}) 
&= \lim_{n \to \infty} 
(\varphi_{S_n,\lambda_n}^{1/2}|\varphi_{T_n}^{1/2})\\
&= \lim_{n \to \infty} 
(\varphi_{S_n}^{1/2}|\varphi_{T_n}^{1/2}) 
e^{-(A_n+B_n)^{-1}(\lambda_n)/2}\\ 
&= (\varphi_S^{1/2}|\varphi_T^{1/2}) 
\lim_{n \to \infty} e^{-(A_n+B_n)^{-1}(\lambda_n)/2}
\end{align*}

Particularly, we observe that 
$(\varphi_{S,\lambda}^{1/2}|\varphi_T^{1/2}) = 0$ if $S$ and $T$ are 
not HS-equivalent, and our task is reduced to showing that 
\[
\lim_{n \to \infty} (A_n+B_n)^{-1}(\lambda_n) = 
\begin{cases}
(A+B)^{-1}(\lambda) &\text{if $\lambda$ is bounded,}\\
+\infty &\text{otherwise}
\end{cases}
\]
under the assumption that $S$ and $T$ are HS-equivalent. 

To deal with this problem, 
we use $R \equiv S+\overline S + T + \overline T$ 
as a reference inner product with $R_n$ restriction 
to $V_n^\C$. 
Let $g_n: V \to V_n$ be the orthogonal projection 
with respect to $R$. 

For $x, y \in V_n^\C$, we have 
\begin{align*}
S_n(x,y) &= S(x,y) = 
R(x,\frac{S}{R}y) 
= R(x,g_n\frac{S}{R}g_ny)\\
&= (S_n + \overline{S_n} + T_n + \overline{T_n})
(x,g_n\frac{S}{R}g_ny), 
\end{align*}
which shows that 
\[
\frac{S_n}{S_n + \overline{S_n} 
+ T_n + \overline{T_n}} 
= \left. 
g_n \frac{S}{S+\overline S + T + \overline T}g_n
\right|_{V_n^\C}. 
\]
Thus we have the expression 
\begin{multline*}
\frac{(S_n^{1/2} + {\overline{S_n}}^{1/2})^2}
{S_n + \overline{S_n} + T_n + \overline{T_n}}
= \frac{S_n + \overline{S_n}}
{S_n + \overline{S_n} + T_n + \overline{T_n}}
\frac{(S_n^{1/2} + {\overline{S_n}}^{1/2})^2}
{S_n + \overline{S_n}}\\
= \left. 
g_n \frac{S + \overline S}{S + \overline S + T + 
\overline T} g_n 
\left( 
\Bigl(e_n \frac{S}{S + \overline S} e_n\Bigr)^{1/2} 
+ \Bigl(e_n \frac{\overline S}{S + \overline S} e_n\Bigr)^{1/2} 
\right)^2
\right|_{V_n^\C},   
\end{multline*}
which reveals that 
\begin{align*}
\lim_{n \to \infty} 
g_n 
&\frac{(S_n^{1/2} + {\overline{S_n}}^{1/2})^2}
{S_n + \overline{S_n} + T_n + \overline{T_n}}
g_n\\  
&\quad= \frac{S + \overline S}{S + \overline S 
+ T + \overline T}
\left(
\left( \frac{S}{S + \overline S} \right)^{1/2}
+ 
\left( \frac{\overline S}{S + \overline S} \right)^{1/2}
\right)^2\\
&\quad= 
\frac{(S^{1/2} + {\overline S}^{1/2})^2}
{S + \overline S + T + \overline T}
\end{align*} 
in strong operator topology. 
Thus, if we set 
\[
C_n = 1 - g_n + g_n 
\frac{(S_n^{1/2} + {\overline{S_n}}^{1/2})^2}
{S_n + \overline{S_n} + T_n + \overline{T_n}}
g_n 
+ g_n 
\frac{(T_n^{1/2} + {\overline{T_n}}^{1/2})^2}
{S_n + \overline{S_n} + T_n + \overline{T_n}}
g_n,  
\]
we have 
\[
\lim_{n \to \infty} C_n = C 
\equiv 
\frac{(S^{1/2} + {\overline{S}}^{1/2})^2 
+ (T^{1/2} + {\overline{T}}^{1/2})^2}
{S + \overline{S} + T + \overline{T}}
\]
again in strong operator topology. 

From the definition, the operator $C_n$ is $R$-positive. 
Since $R_n \leq 2(A_n+B_n) \leq 2R_n$, 
we see $1 \leq C_n \leq 2$. 
Note also that, for $x, y \in V_n$, 
\[
R(x,C_ny) = R_n\Bigl(x,\frac{2(A_n+B_n)}{R_n}y
\Bigr) 
= 2(A_n+B_n)(x,y).
\]

Define $x_n \in V_n$ by 
\[
2(A_n+B_n)(x_n,x) 
= R(x_n,C_nx) = \lambda(x)
\quad
\text{for $x \in V_n$}
\]
and set $y_n = C_nx_n \in V_n$. 
Then 
\[
2^{-1}(A_n+B_n)^{-1}(\lambda_n) = 
2(A_n + B_n)(x_n,x_n) = 
R(x_n,C_nx_n)
\]
and the relation 
$R(y_n,x) = \lambda(x)$ for $x \in V_n$ implies that 
$g_my_n = y_m$ for $m \leq n$. 

If $\lambda$ is bounded with $y_\infty \in V$ defined by 
$R(y_\infty,x) = \lambda(x)$ for $x \in V$, 
then $y_n \to y_\infty$ in norm topology. 
Since $C_n^{-1} \to C^{-1}$ in strong operator topology, 
we see that $x_n = C_n^{-1}y_n \to C^{-1}y_\infty$ in 
norm topology as well. Thus 
\[
\lim_{n \to \infty} 2^{-1}(A_n+B_n)^{-1}(\lambda_n) 
= \lim_{n \to \infty} R(x_n,C_nx_n) 
= R(C^{-1}y_\infty, y_\infty) 
= 2^{-1}(A+B)^{-1}(\lambda).
\]

Assume now that $\lambda$ is not bounded and we shall show that 
\[
\liminf_{n \to \infty} R(x_n,C_nx_n) = +\infty. 
\]
Otherwise, 
\[
\liminf R(x_n,x_n) \leq 
\liminf R(x_n,C_nx_n) < +\infty
\]
and we can find a subsequence $\{ x_{n'} \}$ so that 
$x_{n'}$ converges weakly to some $x_\infty \in V$. 
Then, for any $x \in V$, 
\[
R(x,C_{n'}x_{n'}) = 
R((C_{n'}-C)x,x_{n'}) + R(Cx,x_{n'})
\to R(Cx,x_\infty) = R(x,Cx_\infty)
\]
means $y_{n'} \to Cx_\infty$ weakly and therefore 
$\{ R(y_{n'},y_{n'}) \}$ is bounded by 
the Banach-Steinhauss theorem, 
which contradicts with the unboundedness of $\lambda$, 
concluding that
\[
\lim_{n \to \infty} (A_n+B_n)^{-1}(\lambda_n) = +\infty
\]
if $\lambda$ is not bounded. 

\begin{Theorem}
Let $S$, $T$ be covariance forms on 
a presymplectic vector space $(V,\sigma)$ and 
$\alpha$, $\beta$ be linear functionals of $V$. 
Then 
\[
(\varphi_{S,\alpha}^{1/2}| \varphi_{T,\beta}^{1/2}) 
= \sqrt{
\det\left( 
\frac{2\sqrt{AB}}{A+B}
\right)}
e^{-(A+B)^{-1}(\alpha-\beta)/2}, 
\]
where $2A = (S^{1/2} + {\overline S}^{1/2})^2$ and 
$2B = (T^{1/2} + {\overline T}^{1/2})^2$. 
\end{Theorem}


\begin{thebibliography}{20}
\bibitem{Bo}
V.I.~Bogachev, 
\textit{Gaussian Measures}, Amer.~Math.~Soc., 1998. 
\bibitem{J}
S.~Janson, \textit{Gaussian Hilbert Spaces}, 
Cambridge University Press, 1997.
\bibitem{PW}
W.~Pusz and S.L.~Woronowicz, 
Functional calculus for sesquilinear froms and 
the purification map, 
\textit{Rep.~Math.~Phys.}, 
8(1975), 159--170. 
\bibitem{GMS}
S.~Yamagami, 
Geometric mean of states and transition amplitudes, 
Lett.~Math.~Phys., 84(2008), 123--137. 
\bibitem{GQFS}
\underline{\phantom{S.~Yamagami}}, 
Geometry of quasi-free states of CCR algebras, 
arXiv:0801.0858.  
\end{thebibliography}
\end{document}